\documentclass[a4paper]{article}

\usepackage{amsfonts, amsmath, amssymb, graphicx, mathptmx}

\newcommand{\ExpectMeas}[2]{\mathbb{E}^{#1} \left[ #2 \right]}

\newcommand{\keywords}[1]{{\bf Keywords:} {#1}}
\newcommand{\subclass}[1]{{\bf Mathematics Subject Classification (2000):} {#1}}
\newcommand{\JEL}[1]{{\bf JEL:} {#1}}
\newcommand{\qed}{\hfill$\Box$}
\newtheorem{thm}{theorem}[section]
\newtheorem{theorem}[thm]{Theorem}
\newtheorem{proposition}[thm]{Proposition}
\newtheorem{corollary}[thm]{Corollary}
\newenvironment{proof}{{\bf Proof }}{}

\newcommand{\MakeTitle}{\maketitle\newcommand{\and}{$\cdot$ }}

\title{Maximum Entropy Distributions Inferred from Option Portfolios on an Asset
	\thanks{
		We would like to thank David Chevance, Peter J\"{a}ckel, Yannick
		Malevergne and Wolfgang Scherer for helpful comments and suggestions.
		We would also like to thank the organisers of the WBS 5th Fixed Income
		Conference in Budapest, where we had the opportunity to present some of
		our results.
	}
}

\author{
	Cassio Neri
	\thanks{Lloyds Banking Group, \texttt{cassio.neri@lloydsbanking.com}}
  \and
	Lorenz Schneider
	\thanks{Center for Financial Risks Analysis (CEFRA), EMLYON Business
	  School, \texttt{schneider@em-lyon.com}}
}

\date{13 November 2010}

\begin{document}

\MakeTitle

\begin{abstract}
We obtain the maximum entropy distribution for an asset from call and digital option prices.
A rigorous mathematical proof of its existence and exponential form is given, which can also
be applied to legitimise a formal derivation by Buchen and Kelly \cite{BK}.
We give a simple and robust algorithm for our method and compare our results to theirs.
We present numerical results which show that our approach implies very realistic
volatility surfaces even when calibrating only to at-the-money options.
Finally, we apply our approach to options on the S\&P 500 index.


\subclass{91B24 \and 91B28 \and 91B70 \and 94A17}

\JEL{C16 \and C63 \and G13}
\end{abstract}

\section{Introduction}
\label{Introduction}

The recent market turbulence caused by the credit crunch has exposed in a drastic way the consequences of overconfidence
in financial modelling assumptions. Typically, a financial model, such as the famous Black-Scholes model,
will assume that the price of an asset follows a given stochastic process whose parameters need to be calibrated to market prices.
If a model becomes an accepted standard and most market participants adopt it, problems can occur when assumptions that hold
under normal market conditions are also expected to hold under abnormal ones. An example is the stock market crash of
$1987$, where the volatilities used for pricing at-the-money options were also used for pricing far out-of-the-money put options.
As the market headed downwards, it turned out that the true hedging cost for somebody who had sold such puts was far greater than
the received premium.
Another good example is described in the recent paper \cite{CJS}, where the authors demonstrate for CDOs and CDO$^2$s what can
happen to asset prices when model parameters that are hard to observe or estimate with sufficient accuracy are put to a
true stress test. However, they write: ``The good news is that this mistake can be fixed. For example, a Bayesian approach that explicitly
acknowledges that parameters are uncertain would go a long way towards solving this problem.'' \cite{CJS}

Another well-established way to obtain estimates for such parameters
from observable data, which we will follow here, is via maximum
entropy methods (\cite{AFHS}, \cite{BCM}, \cite{BBM}, \cite{BBCM},
\cite{BBC}, \cite{BK}, \cite{DMY}, \cite{F}, \cite{G1}, \cite{G2}).
Such an estimate ``is the least biased estimate possible on the given information,
i.e., it is maximally noncommittal with regard to missing information.'' \cite{Jay}
For example, the probability distribution over the interval $[0,1]$ which maximises entropy is the uniform distribution.
There is no entropy maximiser for distributions over $\mathbb{R}$.
However, when the mean and variance are specified, the Gauss distribution with these parameters maximises entropy.

We concentrate on the distribution of an asset price at a given time in the future, for which there are some option data.
We develop a highly robust technique to find a Maximum Entropy Distribution (MED) for the asset in case we have call and digital
option prices. The density is obtained by partitioning the range of possible stock prices into buckets,
i.e. the intervals between adjacent strikes given by the option data,
but, in contrast to the Black-Scholes model, making no a priori assumption about the asset's distribution.
Instead, we maximise the Boltzmann-Shannon entropy to obtain a distribution that respects only the given option prices and
is otherwise unbiased. The density can in turn be used to interpolate implied volatilities and,
by repeating this operation for a range of maturities, obtain a volatility surface.
The results agree surprisingly well with observed volatility surfaces from the markets.

Buchen and Kelly (\cite{BK}) have proposed a similar entropy maximisation
method to infer the probability distribution for an asset from call prices.
This maximisation problem corresponds to finding a
set of Lagrange multipliers. In \cite{BCM} the authors write: ``There is a
problem with this type of calculation,'' meaning that the formal Lagrange
multipliers approach is not mathematically rigorous. Using convex programming
arguments they legitimise those calculations. Like \cite{BCM} we legitimise the
results found in \cite{BK}. However, we follow a simpler approach by applying a
result of Csisz\'ar's \cite{C}.

Both \cite{BCM} and \cite{BK} present numerical methods to find the Lagrange
multipliers by solving an $N$-dimensional non-linear problem (where $N$ is the
number of constraints given by call prices). As mentioned in \cite{BK}, in the case
of close strikes the problem can be poorly conditioned. In the present work, we add
$N$ digital prices to our constraints and the resulting numerical problem is
highly simplified: Instead of an $N$-dimensional equation, we need to solve a
one-dimensional problem $F(x)=\lambda$ for $N$ different values of $\lambda$
(the same $F$ though), allowing for easy parallelisation and avoiding any
ill-conditioned problem. Additionally, $F$ is a strictly monotonic function whose
derivative is known analytically. Therefore, the Newton-Raphson method can be
used for excellent speed of convergence and stability. Alternatively to
iterative methods, one may try to find an analytical approximation for $F^{-1}$.

In a nutshell, the advantage we obtain is the localisation of the maximum entropy density into asset price buckets,
in which the functional form is a simple exponential function.
But of course there is a price to pay for this localisation technique, and
the price here is the necessity of an additional constraint for each call option used.
This extra constraint is the price of a digital option at the same strike.

The density in our case differs slightly from the one given by the method in \cite{BK}.
We therefore investigate the differences between them. In both
approaches one can also use information from a so-called prior density, if
available, leading to the concept of relative entropy (also called
$I$-divergence and Kullback-Leibler information number), and we compare the
densities obtained by this method.

After finding the MED we give the expressions for the cumulative distribution function and its inverse.
These formulas involve only arithmetic operations and expo\-nential- and logarithm-functions.
They are therefore very easy to implement and fast to compute. This is a highly useful feature for fast Monte Carlo simulations.

Furthermore, we obtain an analytical formula for the price of a call at a given strike.
By calculating several such prices at different strikes, one can recover the implied volatility smile.

We also include a section in which we calibrate to real market data.
Digital options on the S\&P 500 Index (SPX) and the CBOE Volatility Index (VIX) are traded on the Chicago Board Option Exchange (CBOE),
where they are called {\it binary options}. They are specified such that
``Expiration dates and settlement values are the same as for traditional options'' \cite{CBOE}, which is just what we need in our setup.
We show results for two cases: the first, in which we have CBOE quotes for call {\it and} digital options on the SPX and calibrate to them,
and the second, for a different maturity, in which we only have quotes for call options and therefore have to estimate digital prices from call spreads.
The method we propose here is found to work very well in both cases,
and we compare our results to those obtained by calibrating only to call prices as in \cite{BK}.

\section{The Maximum Entropy Distribution Using Calls and Digitals}
\label{MECD}

\subsection{Maximum Entropy Distribution}

We are given a fixed maturity $T$, strictly increasing strikes
$K_0 = 0, K_1, ..., K_n,$ $K_{n+1} = \infty$, and undiscounted call and digital prices
\begin{equation*}
\tilde{C}_i := C(K_i, T) / DF(0,T), \quad
\tilde{D}_i := D(K_i, T) / DF(0,T)
\end{equation*}
at these strikes. The payoffs of the call and digital options are
given in equations \eqref{ME1} and \eqref{ME2}. $DF(0,T)$ denotes the
discount factor.
Throughout we make the convention $\tilde{C}_{n+1} = \tilde{D}_{n+1} = K_{n+1} \tilde{D}_{n+1} = 0.$

Assuming risk neutral pricing, we will determine a density $g$ for the underlying asset price $S(T)$ which maximises entropy
\begin{equation*}
\label{Entropy}
E(g) := - \int_0^\infty g(x) \ln g(x) dx
\end{equation*}
under the constraints
\begin{equation}
\label{ME1}
\ExpectMeas{g}{(S(T) - K_i)^+} = \tilde{C}_i,
\quad i.e. \quad \int_{K_i}^\infty (x - K_i) g(x) dx = \tilde{C}_i
\end{equation}
and
\begin{equation}
\label{ME2}
\ExpectMeas{g}{ \mathbf{I}_{\{S(T) > K_i\}} } = \tilde{D}_i,
\quad i.e. \quad \int_{K_i}^\infty g(x) dx = \tilde{D}_i
\end{equation}
for all $i=0,...,n$.
In particular, these two constraints for $i=0$ mean that $g$ is a density,
since $\int_0^\infty g(x) dx = \tilde{D}_0 = 1$,
and that the martingale condition
$$
\ExpectMeas{g}{S(T)} = \int_0^\infty x g(x) dx = \tilde{C}_0
$$
is satisfied, since $\tilde{C}_0$ is the forward price of $S$ for time $T$.

From the second constraint it immediately follows that
\begin{equation}
\label{digitalConstraint}
\int_{K_i}^{K_{i+1}} g(x) dx = \tilde{D}_i - \tilde{D}_{i+1} \quad \forall i=0,...,n.
\end{equation}

Looking at a call spread with strikes $K_i, K_{i+1}$ raised to level $K_i$, i.e. a derivative that pays $S(T)$
if $K_i < S(T) < K_{i+1}$ and zero otherwise, we obtain the condition
\begin{equation}
\label{callConstraint}
\int_{K_i}^{K_{i+1}} x g(x) dx = (\tilde{C}_i + K_i \tilde{D}_i) - (\tilde{C}_{i+1} + K_{i+1} \tilde{D}_{i+1}) \quad \forall i=0,...,n.
\end{equation}

We now calculate the density $g$ under the constraints given above.
The purpose of Theorem \ref{LocalGlobalMaximiser} is to show that the local constraints \eqref{digitalConstraint} and \eqref{callConstraint}
are equivalent to the global constraints \eqref{ME1} and \eqref{ME2}. Moreover,
\begin{equation*}
- \int_0^\infty g(x)\ln g(x) dx = \sum_{i=0}^n \left( - \int_{K_i}^{K_{i+1}} g(x)\ln g(x)dx \right),
\end{equation*}
and, thus, we only need to maximise $-\int_{K_i}^{K_{i+1}} g(x)\ln g(x)dx$ subject to \eqref{digitalConstraint}
and \eqref{callConstraint} over each bucket.

Let $\mathcal{M}^+$ be the set of positive Borel-measurable functions defined on $[0, \infty[.$
Define
\begin{equation*}
\mathcal{X}
:= \left\{ g \in \mathcal{M}^+\ \vline \
\int_{K_i}^{\infty} g(x) dx = \tilde{D}_i,
\int_{K_i}^{\infty} (x - K_i) g(x) dx = \tilde{C}_i \; \forall i=0,...,n
\right\}
\end{equation*}
and, for all $i=0,...,n$,
\begin{equation*}
\mathcal{X}_i
:= \left \{ g \in \mathcal{M}^+ \ \vline \
\begin{array}{ll}
\displaystyle\int_{K_i}^{K_{i+1}} g(x) dx &= \tilde{D}_i - \tilde{D}_{i+1},
\\
\\
\displaystyle\int_{K_i}^{K_{i+1}} x g(x) dx &= (\tilde{C}_i + K_i \tilde{D}_i) - (\tilde{C}_{i+1} + K_{i+1} \tilde{D}_{i+1})
\end{array}
\right\}
\end{equation*}

\begin{proposition}
\label{IntersectionX}
$\mathcal{X} = \bigcap_{i=0}^n \mathcal{X}_i.$
\end{proposition}

\begin{proof}
It is straightforward to show this using \eqref{ME1}, \eqref{ME2},
\eqref{digitalConstraint} and \eqref{callConstraint}.
\qed
\end{proof}

For $i = 0,...,n,$ we define
\begin{equation*}
E_i(g) := - \int_{K_i}^{K_{i+1}} g(x) \ln g(x) dx \quad \forall g \in \mathcal{X}_i.
\end{equation*}
\begin{theorem}
\label{LocalGlobalMaximiser}
If $g$ is a maximiser of $E$ on $\mathcal{X}$, then $g$ is a maximiser of $E_i$ on $\mathcal{X}_i$.
Conversely, if $g$ is a maximiser of $E_i$ on $\mathcal{X}_i$ for all $i=0,...,n,$ then $g$ is a
maximiser of $E$ on $\mathcal{X}$.
\end{theorem}
\begin{proof}
Let $g$ be a maximiser of $E$ on $\mathcal{X}$, and let $h \in \mathcal{X}_i$.
Define
\begin{equation*}
\tilde{g} = g \cdot \left(1 - \mathbf{I}_{[K_i, K_{i+1}[}\right) + h \cdot \mathbf{I}_{[K_i, K_{i+1}[}.
\end{equation*}
Since $\tilde{g} = h$ on $[K_i, K_{i+1}[$, we have $\tilde{g} \in \mathcal{X}_i$.
Moreover, for $j \neq i$, we have $\tilde{g} = g$ on $[K_j, K_{j+1}[$, and thus $\tilde{g} \in \mathcal{X}_j$.
It follows from Proposition \ref{IntersectionX} that $\tilde{g} \in \mathcal{X}$. Hence, from the maximality of $E(g)$,
we get $E(g) - E(\tilde{g}) \geq 0$. A simple computation gives $E(\tilde{g}) = E(g) -E_i(g) + E_i(h)$, and therefore
$E_i(g) - E_i(h) = E(g) - E(\tilde{g}) \geq 0$. It follows that $g$ maximises $E_i$ on $\mathcal{X}_i$.

Conversely, suppose that $g$ is a maximiser of $E_i$ on $\mathcal{X}_i$ for all $i=0,...,n.$
Let $h \in \mathcal{X}$. We have
$$
E(g) = \sum_{i=0}^n E_i(g) \geq \sum_{i=0}^n E_i(h) = E(h),
$$
which means that $g$ is a maximiser of $E$ on $\mathcal{X}$.
\qed
\end{proof}

We now give a heuristic way of determining the entropy maximiser, but in the next subsection we also give
a rigorous proof that this is indeed the correct result.
Formally applying the Lagrange multipliers theorem, we conclude that the maximiser has the form
\begin{equation}
\label{alphaBeta}
g(x) = \alpha_i e^{\beta_i x} \quad \text{on} \quad [K_i, K_{i+1}[.
\end{equation}
To see this, define the functionals
\begin{eqnarray*}
{\cal F}(g) &:=& -\int_{K_i}^{K_{i+1}} g(x)\ln g(x) dx \\
{\cal G}(g) &:=& \int_{K_i}^{K_{i+1}} g(x) dx \\
{\cal H}(g) &:=& \int_{K_i}^{K_{i+1}} x g(x) dx
\end{eqnarray*}
and solve the equation
$$
\delta{\cal F}(g) + \lambda_1 \delta{\cal G}(g) + \lambda_2 \delta{\cal H}(g) = 0
$$
for the Fr\^{e}chet derivatives.
It follows that
$$
\int_{K_i}^{K_{i+1}} (\ln g(x) + 1) \delta g(x) dx = \int_{K_i}^{K_{i+1}} (\lambda_1 + \lambda_2 x) \delta g(x) dx.
$$
Therefore, on the interval $[K_i, K_{i+1}[$, we must have
$$
\ln g(x) + 1 = \lambda_1 + \lambda_2 x,
$$
and, introducing $\alpha_i := e^{\lambda_1 - 1}$ and $\beta_i := \lambda_2$, we
obtain \eqref{alphaBeta}.

Using the explicit form of $g$ just found in \eqref{digitalConstraint} and \eqref{callConstraint} gives
\begin{eqnarray}
\label{alpha1}
\alpha_i \int_{K_i}^{K_{i+1}} e^{\beta_i x} dx &=& \tilde{D}_i - \tilde{D}_{i+1}, \\
\label{beta1}
\alpha_i \int_{K_i}^{K_{i+1}} x e^{\beta_i x} dx &=& (\tilde{C}_i + K_i \tilde{D}_i) - (\tilde{C}_{i+1} + K_{i+1} \tilde{D}_{i+1})
\end{eqnarray}
for all $i=0,...,n$. For $i < n$, solving \eqref{alpha1} for $\alpha_i$ and then \eqref{beta1} for $\beta_i$, using integration by parts,
gives
\begin{eqnarray}
\label{alpha2}
\alpha_i
&=&
\beta_i \frac{\tilde{D}_i - \tilde{D}_{i+1}}{e^{\beta_i K_{i+1}} - e^{\beta_i K_i}}, \\
\label{beta2}
\frac{K_{i+1} e^{\beta_i K_{i+1}} - K_i e^{\beta_i K_i}}{e^{\beta_i K_{i+1}} - e^{\beta_i K_i}} - \frac{1}{\beta_i}
&=&
\frac{(\tilde{C}_i + K_i \tilde{D}_i) - (\tilde{C}_{i+1} + K_{i+1} \tilde{D}_{i+1})}{\tilde{D}_i - \tilde{D}_{i+1}}.
\end{eqnarray}
Define
$$
\Theta(\beta; K_i, K_{i+1}) := \frac{K_{i+1} e^{\beta K_{i+1}} - K_i e^{\beta K_i}}{e^{\beta K_{i+1}} - e^{\beta K_i}} - \frac{1}{\beta}.
$$
It follows that
$$
\Theta'(\beta; K_i, K_{i+1}) = \frac{1}{\beta^2} - (K_{i+1}-K_i)^2 \frac{e^{\beta (K_{i+1}+K_i)}}{(e^{\beta K_{i+1}} - e^{\beta K_i})^2}.
$$

\begin{figure}[ht]
\includegraphics[width=\textwidth]{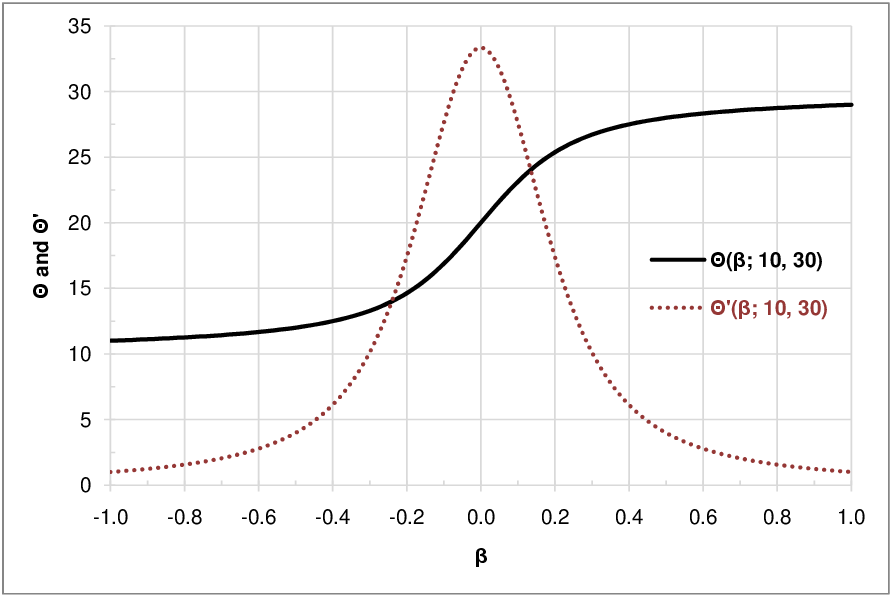}
\caption{Graphs of $\Theta(\beta; 10, 30)$ and $\Theta'(\beta; 10, 30)$ (here $'$ means derivative with respect to $\beta$).}
\label{fig:1}
\end{figure}
Figure \ref{fig:1} shows the graphs of $\Theta(\beta; K_i, K_{i+1})$ and $\Theta'(\beta; K_i, K_{i+1})$ for $K_i = 10$ and $K_{i+1} = 30$.
It suggests that equation \eqref{beta2} has a unique solution if the quantity on the right hand side is in $]K_i, K_{i+1}[$.
This turns out to be the case, as we show with the following proposition.

\begin{proposition}
Let $i \in \{0,...,n\}$. If there is no arbitrage opportunity implied by $\tilde{D}_i, \tilde{D}_{i+1}, \tilde{C}_i, \tilde{C}_{i+1}$,
then there is a unique solution $(\alpha_i, \beta_i)$ for equations \eqref{alpha1} and \eqref{beta1}.
\end{proposition}
\begin{proof}
Define
$$
\bar{K} := \frac{(\tilde{C}_i + K_i \tilde{D}_i) - (\tilde{C}_{i+1} + K_{i+1} \tilde{D}_{i+1})}{\tilde{D}_i - \tilde{D}_{i+1}}.
$$
We first show that we must have $K_i<\bar{K}<K_{i+1}$. This can
be seen by comparing the prices of three derivatives: They pay,
respectively, $K_i$, $S(T)$ and $K_{i+1}$ if $K_i < S(T) < K_{i+1}$
and zero otherwise. Under the assumption that there is no arbitrage
opportunity, it follows immediately that the second derivative is more
expensive than the first one and cheaper that the third one. It is also
clear that they can be replicated by portfolios of calls and digitals and
their prices, in increasing order, are
\[
K_i(\tilde{D_i} - \tilde{D}_{i+1}) < (\tilde{C}_i + K_i\tilde{D}_i) -
(\tilde{C}_{i+1} + K_{i+1}\tilde{D}_{i+1}) < K_{i+1}(\tilde{D_i} -
\tilde{D}_{i+1}).
\]
From the definition of $\bar{K}$ the middle quantity above is
$\bar{K}(\tilde{D_i} - \tilde{D}_{i+1})$ and the result follows.

Next we show that if $K_i < \bar{K} < K_{i+1}$, then there is a unique solution $(\alpha_i, \beta_i)$
for equations \eqref{alpha1} and \eqref{beta1}.
We begin with the case $i < n$. As we have just seen, \eqref{alpha1} and \eqref{beta1} are then equivalent to
\eqref{alpha2} and \eqref{beta2}. Without loss of generality, we may assume $K_i = 0$ and $K_{i+1} = 1$.
Indeed, it is straightforward to see that the change of variables
$$
(\beta, \bar{K}) \leftrightarrow (x, \lambda), \
x := \beta (K_{i+1} - K_i), \
\lambda := \frac{\bar{K} - K_i}{K_{i+1} - K_i}
$$
transforms the equation $\Theta(\beta; K_i, K_{i+1})=\bar{K}$ into $\Theta(x; 0, 1) = \lambda$, with $\lambda \in ]0, 1[$.
Using l'H\^opital's rule we obtain that the function $F: \mathbb{R} \rightarrow \mathbb{R}$ given by
\begin{equation*}
F(x) := \left\{
\begin{array}{cl}
\displaystyle \frac{e^x}{e^x - 1} - \frac{1}{x} \quad & \text{if } x \neq 0, \\
\\
\displaystyle \frac{1}{2} & \text{if } x = 0
\end{array}
\right.
\end{equation*}
is a continuous extension of $\Theta(.; 0, 1)$.
It is easy to see that
$\underset{x \to -\infty}{\lim} F(x) = 0$ and
$\underset{x \to +\infty}{\lim} F(x) = 1$.
Hence the equation $F(x) = \lambda$ has a solution. To prove that the solution is unique, we shall now show that $F$ is strictly increasing.
Again by l'H\^opital's rule, we obtain that $F$ is differentiable at $x=0$ and $F'(0) = 1/12$
(this is particularly useful because $x=0$ is an ideal starting point for the Newton-Raphson method).
For $x \neq 0$ we have
$$
F'(x) = \frac{1}{x^2} - \frac{e^x}{(e^x - 1)^2}.
$$
Recall that $x^{-1} \sinh x > 1$ for all $x \in \mathbb{R}\setminus\{0\}$. Hence
$$
\frac{e^{x/2} - e^{-x/2}}{x} > 1
\Rightarrow \frac{e^x - 1}{x} > e^{x/2}
\Rightarrow \left( \frac{e^x - 1}{x} \right)^2 > e^x
\Rightarrow F'(x) > 0.
$$
We conclude that $F'(x) > 0$ for all $x \in \mathbb{R}$. Therefore $F$ is strictly increasing.
Finally, we consider the case $i=n$. Equations \eqref{alpha1} and \eqref{beta1} then become
\begin{eqnarray*}
\alpha_n \int_{K_n}^\infty e^{\beta_n x} dx &=& \tilde{D}_n, \\
\alpha_n \int_{K_n}^\infty x e^{\beta_n x} dx &=& \tilde{C}_n + K_n \tilde{D}_n.
\end{eqnarray*}
The first equation implies that $\beta_n < 0$.
Solving it for $\alpha_n$ and the second equation for $\beta_n$ gives
$$
\alpha_n = -\frac{\beta_n \tilde{D}_n}{e^{\beta_n K_n}}, \ \beta_n = -\frac{\tilde{D}_n}{\tilde{C}_n}.
$$
\qed
\end{proof}

Note that we have shown that $F$ is itself a continuously differentiable probability distribution function.
For such a function there might already exist an inversion-algorithm.

\subsection{A Rigorous Way of Finding the Entropy Maximiser}
\label{RigorousProof}

Like others, we have formally derived the expression for the entropy maximiser using the Lagrange multipliers method. However, as pointed out in \cite{BCM} ``there is a problem with this type of calculation.'' Recall that the Lagrange multipliers theorem requires continuous differentiability for objective and constraint functionals in a neighbourhood of the maximiser. However, the Boltzmann-Shannon entropy functional is finite only for densities in
\begin{equation*}
{\cal O} := \{g\in L^1(0,\infty)\ |\ g\ln g\in L^1(0,\infty)\},
\end{equation*}
which has empty interior on $L^1(0,\infty)$. Therefore a maximiser is not an interior point of $\cal O$. Even worse, the entropy is far from being continuously differentiable since it is nowhere continuous.

In \cite{BCM}, convex programming arguments are considered to circumvent this problem. Here we present a new approach based on a result by Csisz\'ar \cite{C}.

When no prior density is given we are interested in the (non-relative) entropy of $g$
\begin{equation*}
E_I(g) := -\int_I g(x)\ln g(x)dx,
\end{equation*}
where $I\subset[0,\infty[$ is an interval. However, Csisz\'ar's results deal with relative entropy
of a probability density $g$ with respect to a probability measure $R$ on $I$
\begin{equation*}
E_I(g|R) := -\int_I g(x)\ln g(x)dR(x).
\end{equation*}
Roughly speaking, we are interested in the ``relative entropy'' with respect to the Lebesgue measure which, in general, is not a probability measure.

For $i=0, \dots,n-1$, $I=[K_i, K_{i+1}[$ is bounded. In that case, the problem can easily fit in Csisz\'ar's framework by considering the normalised Lebesgue probability measure $dR(x)=(K_{i+1}-K_i)^{-1}dx$. However, it is impossible to use this trick for the global problem, $I=[0,\infty[$, and for the last bucket, $I=[K_n,\infty[$, since there is no normalisation constant which turns the Lebesgue measure into a probability measure on these intervals.

Nevertheless, it is possible to turn the two problems over unbounded intervals into equivalent ones that do fit in Csisz\'ar's framework.
This is the subject of the next proposition. Moreover, the same arguments also apply to bounded intervals.
Therefore, in contrast to \cite{BCM}, we do not need to make any distinction between bounded and unbounded intervals.

For the sake of simplicity, the statement in the following proposition considers only two {\it main} constraints, namely, the total mass and the mean.
This includes the bucket problems and excludes the global problem (where additional constraints are given).
However, the proof works even for an infinite number of constraints provided the two main ones are among them.

\begin{proposition}
\label{Bijection}
Let $I \subseteq [0, \infty[$ be an interval.
Define $m(x)=\theta e^{-x}$ for all $x\in I$, where $\theta>0$ is a normalisation constant such that $dR(x)=m(x)dx$ is a probability measure on $I$. Let $a_0, a_1 > 0$. Then the mapping $g\mapsto g/m$ is a bijection from
\begin{equation*}
\Omega := \left\{g\in \mathcal{M}^+\ \vline\ \int_I g(x)dx=a_0,\ \int_I xg(x)dx=a_1 \right\}
\end{equation*}
onto
\begin{equation*}
\tilde \Omega := \left\{\tilde g\in \mathcal{M}^+\ \vline\ \int_I \tilde g(x)dR(x)=a_0,\ \int_I x\tilde g(x)dR(x)=a_1\right\}.
\end{equation*}
Moreover, $g$ is a maximiser of $E_I$ on $\Omega$ if and only if $g/m$ is a maximiser of $E_I(\cdot|R)$ on $\tilde\Omega$.
\end{proposition}

\begin{proof}
Define $\Psi: \mathcal{M}^+ \rightarrow \mathcal{M}^+$ by $\Psi(g) := g/m$. Since $m$ is strictly positive, it follows immediately that $\Psi$ is a well defined bijection.

We shall show that $\Psi$ preserves some linear functionals. Let $g\in \mathcal{M}^+$ and $f:I\rightarrow{\mathbb R}$. Then we have
\begin{equation*}
\int_I f(x)\big(\Psi(g)(x)\big)dR(x) = \int_I f(x)\frac{g(x)}{m(x)}dR(x) = \int_I f(x) g(x) dx.
\end{equation*}
In particular, applying this result to $f(x)\equiv 1$ and to $f(x)\equiv x$, it follows immediately that $\Psi$ maps $\Omega$ onto $\tilde\Omega$.

To complete the proof it suffices to show that if $g,h\in\Omega$, then
\begin{equation*}
E_I(g)-E_I(h)\ge 0 \Longleftrightarrow E_I(\Psi(g)|R)-E_I(\Psi(h)|R)\ge 0.
\end{equation*}
In fact, we shall show a stronger result, namely, that the two differences above are equal. This is equivalent to showing that $E_I(\Psi(g)|R) - E_I(g)$ does not depend on $g\in\Omega$. We have
\begin{eqnarray*}
E_I(\Psi(g)|R)
&=&
-\int_I \big( \Psi(g)(x) \big) \ln \big( \Psi(g)(x) \big) dR(x) \\
&=&
-\int_I \frac{g(x)}{m(x)} \ln \left( \frac{g(x)}{m(x)} \right) dR(x) \\
&=&
-\int_I g(x) \ln \left( \frac{g(x)}{m(x)} \right) dx \\
&=&
-\int_I g(x) \ln g(x) dx + \int_I g(x) \ln m(x)dx \\
&=&
E_I(g) +\ln \theta \int_I g(x) dx - \int_I x g(x)dx \\
&=&
E_I(g) + a_0 \ln \theta - a_1.
\end{eqnarray*}
\qed
\end{proof}

Later we will restate and apply a partial version of a theorem by Czisz\'ar. But before we do so, let us say a few words about it.

It is very natural to apply the Lagrange multipliers theorem for maximisation problems under constraints. However, there are many cases where other techniques are used - for instance in the proof of the existence of projection on a convex set of a Hilbert space. In that case, geometric arguments, including the parallelogram identity, are used.

Many texts suggest thinking of the relative entropy of one probability measure with respect to another as a quantity measuring how much they differ. Moreover, they present some similarities between relative entropy and a metric. Unfortunately, they say, this analogy does not go too far. Csisz\'ar's paper pushes these similarities a bit further, showing a relation analogous to the parallelogram identity. Furthermore, he proves the existence of an entropy minimiser%
\footnote{In Csisz\'ar's paper, the minus sign in front of entropy's definition is dropped and its minimisation (rather than maximisation) is studied.}
under convex constraints by similar arguments that show the existence of projection on convex subsets of Hilbert spaces.

We restate here a partial version of his Theorem 3.1 sufficient for our purposes.

\begin{theorem} [Csisz\'ar]
Let $R$ be a probability on a measurable space $(X, {\cal H})$. Let $\{f_\gamma\}_{\gamma\in\Gamma}$ be an arbitrary set of real-valued
$\cal H$-measurable functions on $X$ and $\{a_\gamma\}_{\gamma\in\Gamma}$ be real constants. Let $\cal E$ be
the set of all those probabilities $P$ on $(X,{\cal H})$ for which the integrals $\int f_\gamma dP$ exist and equal $a_\gamma$
$(\gamma\in\Gamma)$. Then, if there exists $Q\in{\cal E}$ such that $Q\ll R$ and its Radon-Nikodym derivative has the form
\begin{equation}
\label{Maximiser}
\frac{\partial Q}{\partial R}(x) = ce^{q(x)}\quad\forall x\in I,
\end{equation}
where $c>0$ and $q$ belongs to the linear space spanned by the $f_\gamma$'s, then
\begin{equation*}
\int \frac{\partial Q}{\partial R}(x)\ln\left(\frac{\partial Q}{\partial R}(x)\right)dR(x)
\le
\int \frac{\partial P}{\partial R}(x)\ln\left(\frac{\partial P}{\partial R}(x)\right)dR(x)
\end{equation*}
for all $P\in{\cal E}$ such that $P\ll R$.
\end{theorem}

Now we prove that $g$, given by \eqref{alphaBeta}, \eqref{alpha2} and \eqref{beta2}, is indeed an entropy maximiser.

\begin{theorem}
Let $i\in{0, \dots, n}$, $I=[K_i, K_{i+1}[$. Let $\alpha_i$ and $\beta_i$ be defined by equations \eqref{alpha2} and \eqref{beta2}.
Then $g:I\rightarrow{\mathbb R}$ given by
\begin{equation*}
g(x) = \alpha_i e^{\beta_i x} \quad \forall x\in I
\end{equation*}
maximises $E_i$ on $\mathcal{X}_i$.
\end{theorem}
\begin{proof}
Set $a_0 = \tilde D_{i+1} - \tilde D_i$ and $a_1 = (\tilde C_i + K_i\tilde D_i) - (\tilde C_{i+1} + K_{i+1}\tilde D_{i+1})$.
Let $m$, $R$, $\Omega$ and $\tilde\Omega$ be as in Proposition \ref{Bijection}.
Note that for this choice of $I$, $a_0$ and $a_1$, we have $\Omega = \mathcal{X}_i$ and $E_I = E_i$.

Let $X=I$, ${\cal H}$ be the $\sigma$-algebra of Lebesgue measurable subsets of $I$, $\Gamma = \{0,1\}$,
$f_\gamma(x)\equiv a_0x^{\gamma}$ ($\gamma\in\Gamma$) and $\cal E$ as in Csisz\'ar's theorem.

Given $\tilde h \in \tilde\Omega$, define the measure $P_{\tilde h}$ by $dP_{\tilde h}(x)=a_0^{-1}\tilde h(x)dR(x)$.
From the definition of $\tilde\Omega$, it follows that $P_{\tilde h}\in{\cal E}$.
Conversely, if $P\in{\cal E}$ and $P\ll R$, then $a_0 \cdot \partial P/ \partial R \in \tilde\Omega$.
Then a simple computation yields
\begin{eqnarray}
\label{Entropy on E}
-\int_I \frac{\partial P_{\tilde h}}{\partial R}(x) \ln \left( \frac{\partial P_{\tilde h}}{\partial R}(x) \right) dR(x)
&=&
-a_0^{-1} \int_I \tilde h(x) \ln \tilde h(x) dR(x) + \ln(a_0) \nonumber \\
&=&
a_0^{-1}E_I(\tilde h|R) + \ln(a_0).
\end{eqnarray}

By definition of $\alpha_i$ and $\beta_i$ we have $g \in \mathcal{X}_i = \Omega$.
Proposition \ref{Bijection} yields $\tilde g=g/m\in\tilde\Omega$.
Moreover, $\tilde g(x) = \alpha_i\theta^{-1} e^{(\beta_i + 1)x}$ for all $x\in I$.

Let $Q=P_{\tilde g}$. It follows that $Q\in{\cal E}$, and its Radon-Nikodym derivative with respect to $R$
(which is $a_0^{-1}\tilde g$) has the form \eqref{Maximiser} with
$c=a_0^{-1}\alpha_i\theta^{-1}$ and $q=(\beta_i+1)a_0^{-1}f_1$.
Therefore, Csisz\'ar's theorem gives
\begin{equation*}
\int \frac{\partial Q}{\partial R}(x)\ln\left(\frac{\partial Q}{\partial R}(x)\right)dR(x)
\le
\int \frac{\partial P}{\partial R}(x)\ln\left(\frac{\partial P}{\partial R}(x)\right)dR(x)
\end{equation*}
for all $P\in{\cal E}$ such that $P\ll R$.
In particular, for all $\tilde h \in \tilde\Omega$, from \eqref{Entropy on E} we obtain
\begin{equation*}
E_I({\tilde g}|R) \ge E_I({\tilde h}|R).
\end{equation*}

We conclude that $\tilde g$ maximises $E_I(\cdot|R)$ on $\tilde\Omega$ and, again by Proposition \ref{Bijection},
that $g$ is a maximiser of $E_i$ on $\mathcal{X}_i$.
\qed
\end{proof}

\subsection{Some Results Regarding the Entropy Maximiser}

We have the explicit form of the density given by equation \eqref{alphaBeta}.
This allows us to give formulas in several important cases. To do this, we
first state two useful results for the following proofs.

For $K\in[K_i,K_{i+1}[$, we have
\begin{eqnarray}
\int_{K_i}^K g(x)dx &=& \alpha_i\int_{K_i}^K e^{\beta_i x}dx = \frac{\alpha_i}{\beta_i}(e^{\beta_iK} - e^{\beta_iK_i}),
\label{K_i-K-mass} \\
\int_{K_i}^K xg(x)dx &=& \alpha_i\int_{K_i}^K xe^{\beta_i x}dx = \frac{d}{d\beta_i}\left[ \alpha_i\int_{K_i}^K e^{\beta_i x}dx \right] \nonumber \\
&=& \frac{\alpha_i}{\beta_i}(Ke^{\beta_iK} - K_ie^{\beta_iK_i}) - \frac{\alpha_i}{\beta_i^2}(e^{\beta_iK} - e^{\beta_iK_i}).
\label{K_i-K-mean}
\end{eqnarray}

It is straightforward to integrate the density $g$ and obtain an explicit form of the
probability distribution
$$
G(x) := \int_0^x g(s) ds.
$$
Its inverse can also be expressed analytically, which is a useful feature for Monte Carlo simulations.
The results are stated in the following proposition.

\begin{proposition}
\label{distributionFunction}
Suppose $K \in [K_i, K_{i+1}[$. Then
\begin{equation*}
G(K) = \left\{
\begin{array}{ll}
\displaystyle
1 - \tilde D_i + \frac{\alpha_i}{\beta_i} (e^{\beta_i K} - e^{\beta_i K_i}) \quad & \mathrm{if}\ \beta_i\ne 0, \\
\\
1 - \tilde D_i + \alpha_i(K - K_i) & \mathrm{if}\ \beta_i = 0. \\
\end{array}
\right.
\end{equation*}

Given $L\in[0,1[$, find $i\in\{0,\dots, n\}$ such that $1-L\in\ ]\tilde D_{i+1}, \tilde D_i]$. Then
\begin{equation*}
G^{-1}(L) = \left\{
\begin{array}{ll}
\displaystyle
\frac1{\beta_i}\ln\left(e^{\beta_iK_i} + \frac{\beta_i}{\alpha_i}(\tilde D_i-1+L)\right) \quad & \mathrm{if} \
\beta_i\ne 0, \\
\\
\displaystyle
K_i + \frac{\tilde D_i -1 + L}{\alpha_i} & \mathrm{if}\ \beta_i = 0.
\end{array}
\right.
\end{equation*}
\end{proposition}

\begin{proof}
We treat only the case $\beta_i\ne 0$. The simpler case $\beta_i=0$ is left to the reader.

First, notice that $G(K_i)=1-\tilde D_i$. Then, using \eqref{K_i-K-mass}, we get
\begin{equation}
\label{pdf}
G(K) = \int_0^{K_i} g(x)dx + \int_{K_i}^K g(x)dx = 1 - \tilde D_i + \frac{\alpha_i}{\beta_i}( e^{\beta_i K} - e^{\beta_i K_i}).
\end{equation}

Since $L\in[1-\tilde D_i,1-\tilde D_{i+1}[\ =[G(K_i),G(K_{i+1})[$, we have $G^{-1}(L)\in[K_i, K_{i+1}[$. Therefore solving \eqref{pdf} for $K=G^{-1}(L)$ concludes the proof.
\qed
\end{proof}

It is also straightforward to express the prices of call and digital options analytically.

\begin{proposition}
\label{Prices}
Given a strike $K\in[0,\infty[$, find $i\in\{0,\dots,n\}$ such that $K\in[K_i,K_{i+1}[$. If $\beta_i\ne 0$, then
\begin{eqnarray*}
\tilde D(K) &=& \tilde D_i - \frac{\alpha_i}{\beta_i}(e^{\beta_iK} - e^{\beta_iK_i}), \\
\tilde C(K) &=& \tilde C_i - (K - K_i) \left(\tilde D_i + \frac{\alpha_i}{\beta_i} e^{\beta_iK_i}\right)+\frac{\alpha_i}{\beta_i^2}(e^{\beta_iK} - e^{\beta_iK_i}).
\end{eqnarray*}

If $\beta_i=0$, then
\begin{eqnarray*}
\tilde D(K) &=& \tilde D_i - \alpha_i(K - K_i), \\
\tilde C(K) &=& \tilde C_i - (K - K_i)\tilde D_i + \frac{\alpha_i}{2}(K - K_i)^2.
\end{eqnarray*}
\end{proposition}

\begin{proof}
Again we prove only the case $\beta_i\ne 0$.
From \eqref{K_i-K-mass} we obtain
\begin{equation}
\label{Digital}
\tilde D(K) = \int_K^\infty g(x)dx = \int_{K_i}^\infty g(x)dx - \int_{K_i}^K g(x)dx = \tilde D_i - \frac{\alpha_i}{\beta_i}(e^{\beta_iK} - e^{\beta_iK_i}).
\end{equation}
For the (undiscounted) call price we have
\begin{eqnarray}
\tilde C(K) + K\tilde D(K) &=& \int_K^\infty xg(x)dx = \int_{K_i}^\infty xg(x)dx - \int_{K_i}^K xg(x)dx \nonumber \\
&=& \tilde C_i + K_i\tilde D_i - \int_{K_i}^K xg(x)dx.
\label{C+KD}
\end{eqnarray}

Now putting \eqref{K_i-K-mean} and \eqref{Digital} into \eqref{C+KD} leads to the stated result.
\qed
\end{proof}

Finally, using Euler's relationship for homogeneous functions, we can also give an explicit formula for spot-delta.

\begin{corollary}
Given a strike $K\in[0,\infty[$, find $i\in\{0,\dots,n\}$ such that $K\in[K_i,K_{i+1}[$. Let $S$ be today's underlying spot price and $\Delta$ be the spot-delta of a call with strike $K$ maturing at $T$. If $\beta_i\ne 0$, then

\begin{equation*}
\Delta = \frac{DF(0,T)}{S}\left(\tilde C_i + K_i\tilde D_i - \frac{\alpha_i}{\beta_i}(Ke^{\beta_iK} - K_ie^{\beta_iK_i}) + \frac{\alpha_i}{\beta_i^2}(e^{\beta_iK} - e^{\beta_iK_i}) \right).
\end{equation*}

If $\beta_i=0$, then
\begin{equation*}
\Delta = \frac{DF(0,T)}{S}\left(\tilde C_i + K_i\tilde D_i - \frac{\alpha_i}{2}(K^2-K_i^2)\right).
\end{equation*}
\end{corollary}

\begin{proof}
Again we consider only the case $\beta_i\ne 0$ and leave the simpler case $\beta_i=0$ for the reader.

Let $C$ and $D$ be, respectively, the discounted prices of call and digital options with strike $K$ maturing at $T$, i.e. $C = DF(0,T)\tilde C(K)$ and $D = DF(0,T)\tilde D(K)$, where $\tilde C$ and $\tilde D$ are as in Proposition \ref{Prices}.

Since $C$ is a positively homogeneous function of degree $1$ in $(K,S)$, from Euler's theorem we have
\begin{equation*}
C = S \frac{\partial C}{\partial S} + K \frac{\partial C}{\partial K}.
\end{equation*}

Recalling that $D=-\frac{\partial C}{\partial K}$, we can rewrite the last relation as
\begin{equation*}
S\Delta = C + KD = DF(0,T)(\tilde C(K) + K\tilde D(K)).
\end{equation*}
Now using \eqref{K_i-K-mean} and \eqref{C+KD} gives the result.
\qed
\end{proof}

Note that the analogous statement for the forward-delta can be obtained by replacing the spot-price
with the forward-price in the corollary and proof above.

\subsection{Maximum Relative Entropy Distribution with a Given Prior Distribution}

If we hold a prior belief about the distribution, we can maximise relative entropy instead
in order to stay as ``close'' as possible to the prior distribution.
Suppose $p(x)$ is a probability density for this prior distribution.
For ${\mathbb P}_h \ll {\mathbb P}_p$, define relative entropy
\begin{equation*}
E(h|p) := - \int_0^\infty h(x) \ln \left( \frac{h(x)}{p(x)} \right) dx.
\end{equation*}
(The Kullback-Leibler information number or $I$-divergence is given by $E_{KL}(h|p) = -E(h|p)$.
This can be thought of as a measure of distance between two distributions.
For example, $E_{KL}(h|p) \geq 0 \ \forall h$, and $E_{KL}(h|p) = 0$ if and only if $h = p$.)
We have
$$
- \int_0^\infty h(x) \ln \left( \frac{h(x)}{p(x)} \right) dx = - \int_0^\infty \frac{h(x)}{p(x)} \ln \left( \frac{h(x)}{p(x)} \right) p(x) dx,
$$
and essentially the same argument as the one given above shows that the Maximum Relative Entropy Density
(MRED) $h$ is given by
\begin{equation}
\label{relativeEntropyDensity}
\frac{h(x)}{p(x)} = \gamma_i e^{\delta_i x} =: g(x), \quad x \in [K_i, K_{i+1}[.
\end{equation}
Therefore the resulting density $h = g p$ is now given by the product of a piecewise exponential density
and the prior density.

Even in the simple case where the prior density $p$ is just log-normal, we no longer have explicit formulas for
call and digital prices. Since we cannot separate the two constraints
\begin{eqnarray*}
\label{bothConstraints}
\int_{K_i}^{K_{i+1}} h(x) dx &=& \tilde{D}_i - \tilde{D}_{i+1},
\\
\int_{K_i}^{K_{i+1}} x h(x) dx &=& (\tilde{C}_i + K_i \tilde{D}_i) - (\tilde{C}_{i+1} + K_{i+1} \tilde{D}_{i+1})
\end{eqnarray*}
for each $i=0,...,n,$ as in equations \eqref{alpha1} and \eqref{beta1},
we must solve them simultaneously using numerical integration and a two-dimensional root-finder.

However, if the prior density $p$ is already given by an MED, then
$$
h(x) = g(x) p(x) = \alpha_i \gamma_i e^{(\beta_i + \delta_i) x}, \quad x \in [K_i, K_{i+1}[,
$$
and we can solve everything analytically as before. We also recover explicit formulas for call and digital prices.

\section{The Maximum Entropy Distribution Using Calls}
\label{MEC}

\subsection{Maximum Entropy Distribution}

Buchen and Kelly \cite{BK} propose a similar method to find an entropy-maximising density $g_{BK}$ under constraints
given by European payoffs.
The case of most interest is where these are the payoff-functions of call options at different strikes $K_1, ..., K_m$
and the actual constraints are given by (undiscounted) call option prices $\tilde{C}_1, ..., \tilde{C}_m$
such that
$$
\ExpectMeas{g_{BK}}{ \left( S(T)-K_i \right)^+} = \tilde{C}_i
$$
must hold for all $i=1,...,m$.

The density $g_{BK}$ must therefore satisfy the conditions
\begin{equation}
\label{BK-ME1}
\int_0^\infty (x - K_i)^+ g_{BK}(x) dx = \tilde{C}_i \quad \forall i=1,...,m
\end{equation}
and
\begin{equation}
\label{BK-ME2}
\int_0^\infty g_{BK}(x) dx = 1.
\end{equation}

To find $g_{BK}$, they construct the functional
\begin{eqnarray*}
{\cal H}(g_{BK})
&:=& -\int_0^{\infty} g_{BK}(x) \, \ln g_{BK}(x) dx
\\
&+& (1 + \lambda_0) \int_0^{\infty} g_{BK}(x) dx
\\
&+& \sum_{i=1}^m \lambda_i \int_0^{\infty} (x - K_i)^+ g_{BK}(x) dx,
\end{eqnarray*}
where $\lambda_0, ..., \lambda_m$ are the Lagrange multipliers, and then solve the equation
\begin{equation*}
\delta{\cal H}
= \int_0^{\infty} \left( -\ln g_{BK}(x) + \lambda_0 + \sum_{i=1}^m \lambda_i (x - K_i)^+ \right) \delta g_{BK}(x) dx = 0.
\end{equation*}
The solution is given by
\begin{equation}
\label{BK-density}
g_{BK}(x) = \frac{1}{\mu} e^{\sum_{i=1}^m \lambda_i (x - K_i)^+} \quad \forall x \in [0, \infty[,
\end{equation}
where $\mu := e^{-\lambda_0} = \int_0^{\infty} e^{\sum_{i=1}^m \lambda_i (x - K_i)^+} dx$ is a normalising constant.

Buchen and Kelly show that numerically, finding the parameters $\lambda_1, ..., \lambda_m$ is an m-dimensional root-finding problem
that can be tackled with the multi-dimensional Newton algorithm.
They show how to compute the Jacobian, and that it is invertible, by expressing it as a covariance matrix.

If a call option with strike $K_1 = 0$, i.e. the forward, is among the input data, the mean of the distribution is given.
Since the total mass, $1$, is also known, we have the two main constraints needed to apply the arguments
from subsection \ref{RigorousProof} and can therefore also rigorously find the entropy maximiser when
only call options are given as input. Of course, the forward should be known in most situations, so that this
is certainly the most important case.

\subsection{Maximum Entropy Distribution with a Given Prior Distribution}

Similarly, if a prior distribution $p$ is given, the distribution maximising relative entropy under the same constraints
is given by
\begin{equation}
h_{BK}(x) = \frac{p(x)}{\mu} e^{\sum_{i=1}^m \lambda_i (x - K_i)^+} \quad \forall x \in [0, \infty[,
\end{equation}
where $\mu = \int_0^{\infty} p(x) e^{\sum_{i=1}^m \lambda_i (x - K_i)^+} dx$ is again the normalising constant.

\section{Comparing the Distributions}
\label{ComparingDistributions}

In this section we give some numerical examples for the entropy maximisers described so far.
We suppose that the market data is given by
$$
F = 100, \ r = 0, \ \sigma = 25\%, \ T = 1.
$$
We assume a flat volatility and make no skew correction when calculating the digital prices in this scenario.

\subsection{MED using Calls and Digitals}
\label{MED-CD}

We calculate three densities using strikes
\begin{itemize}
\item
$K_0 = 0, K_1 = 100$
\item
$K_0 = 0, K_1 = 60, K_2 = 100, K_3 = 140$
\item
$K_0 = 0, K_1 = 60, K_2 = 80, K_3 = 100, K_4 = 120, K_5 = 140$
\end{itemize}

\begin{table}[ht]
\scriptsize
\caption{Option Prices and Density Parameters}
\label{tab:1}
\begin{tabular}{lllllllllll}
\hline\noalign{\smallskip}
{\bf Market} \\
Strike & 0.00 &	20.00 & 40.00 & 60.00 & 80.00 &	100.00 & 120.00	& 140.00 & 160.00 &	180.00 \\
Call &	100.0000 & 80.0000 & 60.0005 & 40.1454 & 22.2656 & 9.9477 & 3.7059 & 1.2139 & 0.3659 & 0.1049 \\
Digital & 1.0000 & 1.0000 &	0.9998 & 0.9725 & 0.7786 & 0.4503 &	0.1965 & 0.0707	& 0.0225 & 0.0066 \\
\noalign{\smallskip}\hline\noalign{\smallskip}
{\bf MED for 1 strike} \\	
Entropy & 4.6714 \\
$\alpha$ & 1.3582E-04 & n/a & n/a & n/a & n/a & 1.8835 & n/a & n/a	& n/a & n/a \\
$\beta$ & 0.0539 & n/a & n/a & n/a & n/a & -0.0453 & n/a & n/a & n/a & n/a \\
Call & 100.0000 & 80.0402 & 60.2562 & 40.9886 & 23.2384 & 9.9477 & 4.0232 & 1.6271 & 0.6581 & 0.2661 \\
Implied Vol. & n/a & 62.13\% & 46.26\% & 36.17\% & 28.88\% & 25.00\% & 25.95\% & 27.04\% & 27.84\% & 28.41\% \\
Digital & 1.0000 & 0.9951 & 0.9808 & 0.9386 & 0.8146 & 0.4503 & 0.1821 & 0.0736 & 0.0298 & 0.0120 \\
\noalign{\smallskip}\hline\noalign{\smallskip}
{\bf MED for 3 strikes} \\
Entropy & 4.6143 \\
$\alpha$ & 6.0682E-08 & n/a & n/a & 0.0016 & n/a & 0.5397 & n/a & 14.2333 & n/a & n/a \\
$\beta$ & 0.1894 &	n/a & n/a &	0.0255 & n/a & -0.0343 & n/a & -0.0582 & n/a & n/a \\
Call & 100.0000 & 80.0001 & 60.0033 & 40.1454 & 22.4905 & 9.9477 &	3.7539 & 1.2139 & 0.3790 & 0.1183 \\	
Implied Vol. &	n/a	& 38.76\% & 28.60\% &	25.00\% & 25.93\% & 25.00\% & 25.14\% &	25.00\% & 25.15\% & 25.38\% \\	
Digital & 1.0000 & 1.0000 & 0.9994 &	0.9725 & 0.7765 & 0.4503 & 0.1978 &	0.0707 & 0.0221	& 0.0069 \\	
\noalign{\smallskip}\hline\noalign{\smallskip}
{\bf MED for 5 strikes} \\
Entropy & 4.6076 \\
$\alpha$ & 6.0682E-08 & n/a & n/a & 1.5393E-04 & 0.0129 & 0.2389 & 1.6987 & 14.2333 & n/a & n/a \\
$\beta$ & 0.1894 & n/a & n/a & 0.0584 & 0.0027 & -0.0268 & -0.0433 & -0.0582 & n/a & n/a \\
Call & 100.0000	& 80.0001 &	60.0033 & 40.1454 &	22.2656 & 9.9477 & 3.7059 & 1.2139 & 0.3790 & 0.1183 \\
Implied Vol. & n/a	& 38.76\% & 28.60\% & 25.00\% & 25.00\% & 25.00\% & 25.00\% & 25.00\% & 25.15\% & 25.38\% \\
Digital & 1.0000 & 1.0000 & 0.9994 & 0.9725 & 0.7765 & 0.4503 &	0.1978 & 0.0707	& 0.0221 & 0.0069 \\
\noalign{\smallskip}\hline
\end{tabular}
\end{table}
Table \ref{tab:1} gives the (undiscounted) option prices we used and the parameters describing the density.

\begin{figure}[ht]
\includegraphics[width=\textwidth]{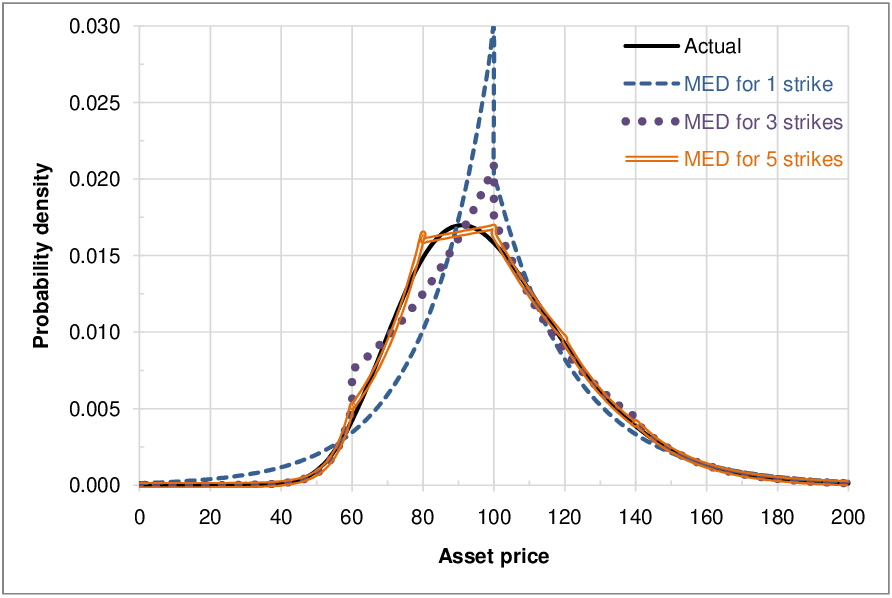}
\caption{Graphs of the actual log-normal density with $\sigma=25\%$ and three maximum entropy densities obtained by calibrating to $1$, $3$ and $5$ strikes.}
\label{fig:2}
\end{figure}
Figure \ref{fig:2} shows the three densities and the actual log-normal density. It can be seen that already with $5$ strikes and the forward,
the fit of the piecewise-exponential distribution to the log-normal distribution is very good.

In practice, however, implied volatilities are not flat as in the example above, i.e. they are not the same for different strikes and maturities.
This is discussed in detail in Gatheral's book \cite{G}, in particular in the section ``The SPX Implied Volatility Surface'' in chapter 3.
We show that our method has by its nature a tendency to give good fits to observed volatility surfaces. To do this, we will assume that we now only
have at-the-money (ATM) option prices, and that already with this minimal amount of market data our method generates a very realistic looking volatility surface.

We show in Figure \ref{fig:3} the implied volatility surface obtained by using just the ATM strike.
More precisely, if we assume $r = 0$ and $F = 100$ again, then the ATM strike is $K_{\rm ATM} = 100$.
Moreover, we consider a constant ATM volatility $\sigma_{\rm ATM} = 25\%$ (it could of course be time dependent).
For each maturity $T > 0$, we compute the Black-Scholes prices of the ATM call $C_1$ and ATM digital $D_1$.
Then applying the maximum entropy approach to just one strike $K_1 = K_{\rm ATM}$, we compute $\alpha_i$ and $\beta_i$ for $i = 0, 1$.
In other terms, we recover the maximum entropy density of $S(T)$ compatible with $K_1$, $C_1$ and $D_1$.
We emphasise that no other strike, call or digital is used in this calibration.

According to the MED, the price of a call with strike $K$ and maturity $T$ is given by $C_{\rm MED}(K, T) = \tilde C(K)$,
where $\tilde C(K)$ is given in Proposition \ref{Prices}.
From $C_{\rm MED}(K,T)$ we recover the implied volatility $\sigma(K, T)$ with a bisection root-finder from the Black-Scholes formula.
Readers interested in a more robust method can consult the one proposed in \cite{Jae}.

As expected, as a consequence of calibration, $\sigma(K_{\rm ATM}, T) = \sigma_{\rm ATM}$.
What is surprising is the fact that the curve $K \mapsto \sigma(K, T)$ has a profile very similar to smile curves typically seen in equity markets.

Now, by varying $T$ one constructs a volatility surface which, again, is qualitatively very similar to those observed in equity markets.

\begin{figure}[ht]
\includegraphics[width=\textwidth]{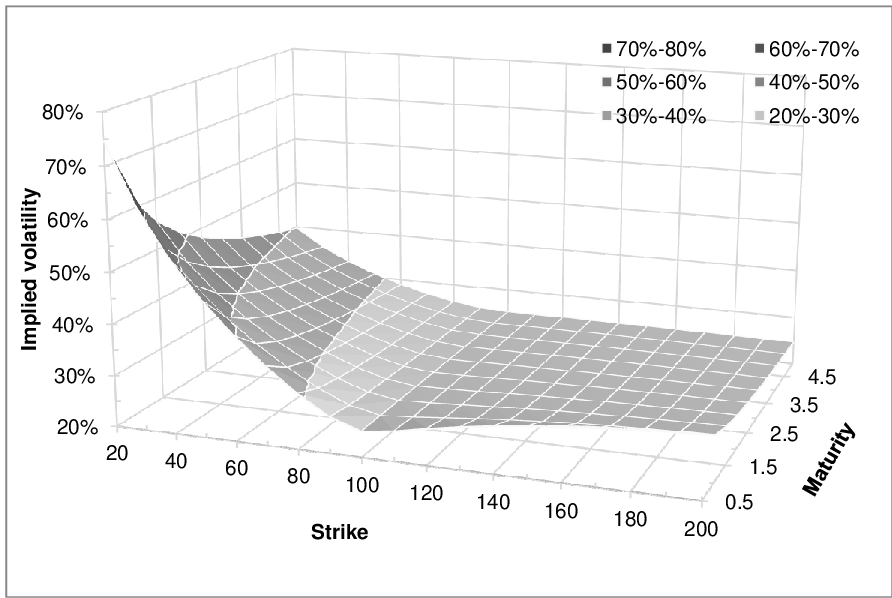}
\caption{Implied volatility surface obtained by calibrating only to the (constant) at-the-money volatility curve.}
\label{fig:3}
\end{figure}

The maximum entropy method seems to be able to transform just one volatility number from a flat Black-Scholes
(ATM) world into a very realistic looking volatility surface, with important features such as a strongly pronounced smile at the short end
that decays as the maturity increases.

Different strikes and different numbers of option prices can of course be used at different maturities, so that any arbitrage-free
option data can easily be converted into an implied volatility surface.

The density $g$ is usually discontinuous at the $K_i$'s. The distribution function is of course continuous.
Many Monte Carlo models work by drawing a random uniform variable and inverting the distribution.
In Black-Scholes type models, for example, a normal distribution has to be inverted at some stage.
In our case, only one logarithm needs to be taken, a circumstance which accelerates a simulation.

\subsection{MED using Calls and Digitals with Prior Log-Normal Distribution}
\label{RelMED-CD}

Let the prior distribution be a log-normal distribution with fixed volatility parameter $\sigma$
$$
p(x) = \frac{1}{x} \frac{1}{\sqrt{2 \pi \sigma^2 T}} e^{-\frac{\left( \ln(x/F) + \sigma^2 T/2 \right)^2}{2 \sigma^2 T}} \quad \forall x \in [0, \infty[.
$$
We still have an explicit form of the density, namely $h = g p$, where $g$ is a piecewise exponential density,
although the parameters $\gamma_i, \delta_i$ are of course different from the parameters $\alpha_i, \beta_i$ used for the MED
of the previous subsection. Since we are now unable to express call prices analytically, we calculate them via numerical integration.

\begin{table}[ht]
\scriptsize
\caption{Maximum Relative Entropy Density Parameters}
\label{tab:2}
\begin{tabular}{lllllll}
\hline\noalign{\smallskip}
Strike & 0.00 &	60.00 &	80.00 &	100.00 & 120.00 & 140.00 \\
\noalign{\smallskip}\hline\noalign{\smallskip}				
{\bf MRED for 1 strike} \\
$\gamma$ & 12.2600 & n/a & n/a & 0.0833 & n/a & n/a \\
$\delta$ & -0.0298 & n/a & n/a & 0.0206 & n/a & n/a \\
\noalign{\smallskip}\hline\noalign{\smallskip}			
{\bf MRED for 3 strikes} \\
$\gamma$ & 11.2900 & 7.2379 & n/a & 0.2930 & n/a & 0.3267 \\
$\delta$ & -0.0194 & -0.0237 & n/a & 0.0098 & n/a & 0.0116 \\
\noalign{\smallskip}\hline\noalign{\smallskip}							
{\bf MRED for 5 strikes} \\
$\gamma$ & 11.2900 & 5.9910 & 2.0430 & 0.5970 & 0.5815 & 0.3267 \\
$\delta$ & -0.0194 & -0.0210 & -0.0097 & 0.0031 & 0.0047 & 0.0116 \\
\noalign{\smallskip}\hline
\end{tabular}
\end{table}
Table \ref{tab:2} gives the parameters describing the density. Of course, should a prior density already meet the constraints,
we will have $\gamma_i = 1$ and $\delta_i =0$ for all $i$.

\begin{figure}[ht]
\includegraphics[width=\textwidth]{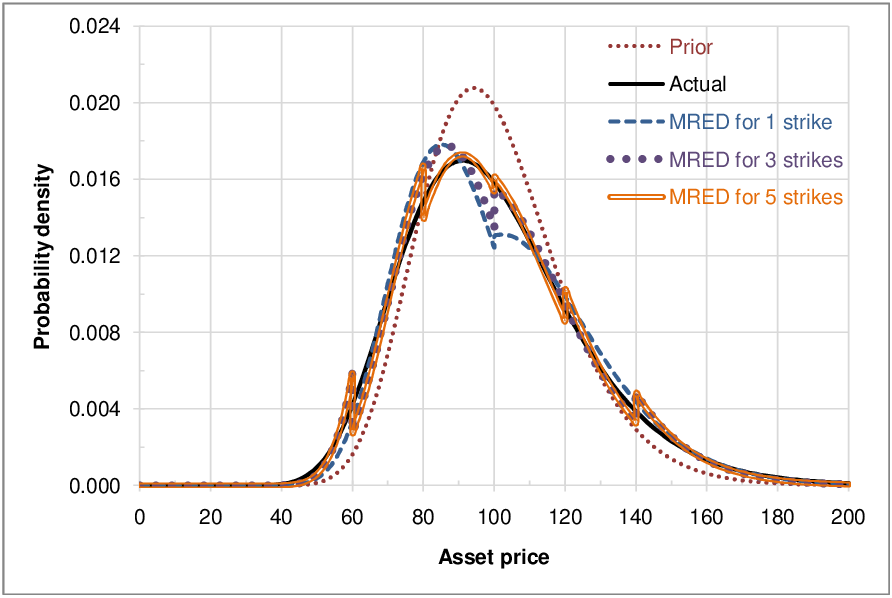}
\caption{Graphs of the prior log-normal density with $\sigma=20\%$, actual log-normal density with $\sigma=25\%$ and three maximum relative entropy densities obtained by calibrating to $1$, $3$ and $5$ strikes.}
\label{fig:4}
\end{figure}
Figure \ref{fig:4} shows the three maximum relative entropy densities and the prior log-normal density ($F = 100, \sigma = 20\%$).
The density for the forward and one call is already much closer to the actual one than in the previous case,
so that convergence is not as pronounced as before when the number of strikes is increased.
We see that $g$ has the effect of pushing the prior density downwards and widening it as to be closer to the actual density.

\subsection{MED using Calls}
\label{MED-C}

The explicit form of the density given by equation \eqref{BK-density} allows one to obtain
analytic expressions for call and digital prices like those in Proposition \ref{Prices}.
As an example, using just the forward and an at-the-money call, i.e. $K_1 = 0, K_2 = 100$, we obtained
$$
\lambda_1 = 0.048747,\ \lambda_2 = -0.098626,\ \mu = 5290.62
$$
on our computer. This leads to a very similar volatility smile as the one given at $T = 1$ in subsection \ref{MED-CD} above.
We refer to \cite{BK} for graphs and numerical data regarding this distribution.

\subsection{MED using Calls with Prior Log-Normal Distribution}
\label{RelMED-C}

As in subsection \ref{RelMED-CD}, in general there will be no analytic expressions for call or digital prices.
If the chosen prior distribution is continuous, then the resulting relative entropy maximiser will
also be continuous.
Again, we advise the reader to look at \cite{BK} for graphs and numerical data regarding this distribution.

\section{Calibrating to CBOE Option Data for the SPX}
\label{SPX-Calibration}

Digital options are traded on the Chicago Board Option Exchange (CBOE). They are called {\it binary options} there.
We quote the following paragraph from the ``Binaries'' product description \cite{CBOE}:

``CBOE offers Binary Options on the S\&P 500 Index (SPX) and the CBOE Volatility Index (VIX).
The ticker symbols for these Binary contracts is BSZ and BVZ respectively.
Expiration dates and settlement values are the same as for traditional options.''

The specification that digital option strikes and maturities are the same as those of call options is exactly what we need for our setup.
We calibrate to CBOE option prices from 10 April 2010 for two different maturities.

\subsection{Call and Digital Option Data for 18 September 2010}
\label{SPX-18Sep2010}

The first maturity is 18 September 2010.
We have digital option bid and ask quotes for strikes $K$ from $950$ to $1400$ USD, usually in steps of $25$ USD.
We also have call option bid and ask quotes for these same strikes.
We calibrate to the ``mid'' prices, i.e. the average of the bid and ask quotes, at the ten strikes from $950$ to $1400$ in steps of $50$ USD.

\begin{table}[ht]
\scriptsize
\caption{Calibrating to CBOE Quotes for Call and Digital Options on the SPX}
\label{tab:3}
\begin{tabular}{rrrrrrr}
\hline\noalign{\smallskip}
        & \multicolumn{2}{c}{\bf Market}   & \multicolumn{2}{c}{\bf MED using Calls}	& \multicolumn{2}{c}{\bf MED using Calls \& Digitals}	\\
Strike & Digital & Call  & Digital & Call    & Digital & Call \\
\noalign{\smallskip}\hline\noalign{\smallskip}
950	 &  0.9400 &	246.30 &	0.9259 &	246.30 &	0.9400 &	246.30 \\
975	 &  0.9150 &	223.20 &	0.9171 &	223.25 &	0.9153 &	223.12 \\
1000 &	0.8950 &	200.50 &	0.9014 &	200.50 &	0.8950 &	200.50 \\
1025 &	0.8750 &	178.15 &	0.8787 &	178.24 &	0.8795 &	178.30 \\
1050 &	0.8550 &	156.60 &	0.8516 &	156.60 &	0.8550 &	156.60 \\
1075 &	0.8150 &	135.70 &	0.8191 &	135.70 &	0.8195 &	135.65 \\
1100 &	0.7750 &	115.70 &	0.7802 &	115.70 &	0.7750 &	115.70 \\
1125 &	0.7250 &	 96.75 &	0.7336 &	 96.76 &	0.7367 &	 96.76 \\
1150 &	0.6700 &	 79.10 &	0.6776 &	 79.10 &	0.6700 &	 79.10 \\
1175 &	0.6050 &	 63.00 & 	0.6117 &	 62.97 &	0.6137 &	 63.01 \\
1200 &	0.5350 &	 48.60 &	0.5357 &	 48.60 &	0.5350 &	 48.60 \\
1225 &	0.4550 &	 36.25 &	0.4541 &	 36.23 &	0.4585 &	 36.13 \\
1250 &	0.3550 &	 25.90 &	0.3720 &	 25.90 &	0.3550 &	 25.90 \\
1300 &	0.1850 &	 11.35 &	0.2112 &	 11.35 &	0.1850 &	 11.35 \\
1350 &	0.0700 &	  4.10 &	0.0896 &	  4.10 &	0.0700 &	  4.10 \\
1400 &	0.0450 &	  1.33 &	0.0307 &	  1.33 &	0.0450 &	  1.33 \\
\noalign{\smallskip}\hline
\end{tabular}
\end{table}
Table \ref{tab:3} shows CBOE prices for call and digital options on the SPX from 10 April 2010 in columns 2 and 3.
Columns 4 and 5 show option prices obtained by calibrating an MED to call prices at strikes $950$, $1000$, $...$, $1350$, $1400$.
Columns 6 and 7 show option prices obtained by calibrating an MED to call and digital prices at the same strikes.
Note that the second MED matches market call {\it and} digital prices at the strikes calibrated to exactly, whereas the
first MED matches only the call prices.

\subsection{Call Option Data for 31 December 2010}
\label{SPX-31Dec2010}

The second maturity is 31 December 2010.
We have call option bid and ask quotes for strikes $K$ from $500$ to $1600$ USD in steps of $50$ USD.
We do not have any digital option quotes for this maturity.
As a substitute, we calculate symmetric call spread prices
\begin{equation}
\label{CallSpread}
\tilde{D}_i = -\frac{\tilde{C}_{i+1} - \tilde{C}_{i-1}}{K_{i+1} - K_{i-1}}
\end{equation}
using mid call prices.

We calibrate to this data at the three strikes $700,\ 1200,\ 1400$. In the previous example, we showed that our method can be used
to calibrate to quotes at many ($10$) strikes, and that this leads to a very good fit.
In this example, we calibrate to only a small number ($3$) of strikes in order to show that our method has an excellent
natural tendency to fit a market smile.
The call spread prices $\tilde{D}_i$ needed for \eqref{CallSpread} are obtained by using the call quotes at $K_i \pm 50$.

\begin{table}[ht]
\scriptsize
\label{tab:4}
\caption{Calibrating to CBOE Quotes for Call Options on the SPX}
\begin{tabular}{rrrrrrrrr}
\hline\noalign{\smallskip}
        & \multicolumn{2}{c}{\bf Market} & \multicolumn{2}{c}{\bf MED using Calls \& Digitals} & \multicolumn{2}{c}{\bf MED using Calls}	& \multicolumn{2}{c}{\bf MED using Calls}	\\
        & \multicolumn{2}{c}{ }          & \multicolumn{2}{c}{\bf for 3 strikes}                & \multicolumn{2}{c}{\bf for 3 strikes}    & \multicolumn{2}{c}{\bf for 9 strikes}	\\
Strike & Call & Implied Vol. & Call & Implied Vol. & Call & Implied Vol. & Call & Implied Vol. \\
\noalign{\smallskip}\hline\noalign{\smallskip}
500   & 681.15 & 62.90\% & 678.87 & 59.45\% & 677.05 & 56.22\% & 679.66 & 60.70\% \\
550   & 631.75 & 57.19\% & 630.27 & 55.25\% & 628.62 & 52.84\% & 630.92 & 56.12\% \\
600   & 582.35 & 51.94\% & 581.72 & 51.22\% & 580.41 & 49.63\% & 582.18 & 51.75\% \\
650   & 533.45 & 47.55\% & 533.22 & 47.32\% & 532.44 & 46.53\% & 533.45 & 47.55\% \\
700   & 484.75 & 43.52\% & 484.75 & 43.52\% & 484.75 & 43.52\% & 484.75 & 43.52\% \\
\hline
750   & 436.55 & 39.99\% & 436.54 & 39.98\% & 437.36 & 40.59\% & 436.55 & 39.99\% \\
800   & 388.85 & 36.77\% & 388.90 & 36.81\% & 390.39 & 37.75\% & 388.97 & 36.85\% \\
850   & 341.85 & 33.85\% & 342.02 & 33.95\% & 343.97 & 35.01\% & 342.11 & 34.00\% \\
900   & 295.95 & 31.28\% & 296.11 & 31.36\% & 298.32 & 32.39\% & 296.20 & 31.40\% \\
950   & 251.25 & 28.89\% & 251.49 & 28.99\% & 253.72 & 29.88\% & 251.55 & 29.01\% \\
1000  & 208.25 & 26.70\% & 208.54 & 26.80\% & 210.56 & 27.50\% & 208.54 & 26.80\% \\
1050  & 167.50 & 24.68\% & 167.76 & 24.76\% & 169.36 & 25.25\% & 167.72 & 24.75\% \\
1100  & 129.70 & 22.82\% & 129.84 & 22.86\% & 130.87 & 23.15\% & 129.77 & 22.84\% \\
1150  & 95.60 & 21.09\% & 95.64 & 21.10\% & 96.06 & 21.21\% & 95.60 & 21.09\% \\
1200  & 66.30 & 19.52\% & 66.30 & 19.52\% & 66.30 & 19.52\% & 66.30 & 19.52\% \\
\hline
1250  & 42.70 & 18.13\% & 43.09 & 18.24\% & 42.93 & 18.20\% & 42.70 & 18.13\% \\
1300  & 25.10 & 16.91\% & 25.83 & 17.13\% & 25.61 & 17.06\% & 25.19 & 16.93\% \\
1350  & 13.35 & 15.87\% & 13.79 & 16.05\% & 13.63 & 15.99\% & 13.35 & 15.87\% \\
1400  & 6.35  & 15.01\% & 6.35  & 15.01\% & 6.35  & 15.01\% & 6.35  & 15.01\% \\
\hline
1450  & 2.68  & 14.27\% & 2.74  & 14.34\% & 2.83  & 14.43\% & 2.68  & 14.27\% \\
1500  & 1.13  & 13.91\% & 1.18  & 14.02\% & 1.26  & 14.16\% & 1.08  & 13.82\% \\
1550  & 0.43  & 13.57\% & 0.51  & 13.88\% & 0.56  & 14.05\% & 0.43  & 13.59\% \\
1600  & 0.20  & 13.69\% & 0.22  & 13.83\% & 0.25  & 14.03\% & 0.17  & 13.49\% \\
\noalign{\smallskip}\hline
\end{tabular}
\end{table}
Table \ref{tab:4} shows CBOE prices for call options on the SPX from 10 April 2010 and their implied volatilities in columns 2 and 3.
Columns 4 and 5 show call option prices and implied volatilities obtained by calibrating an MED to call prices at strikes $700$, $1200$, $1400$.
Columns 6 and 7 show call option prices and implied volatilities obtained by calibrating an MED to call and digital prices at the same strikes, using call spread prices at $700\pm50$, $1200\pm50$, $1400\pm50$ as substitutes for the digital prices.
Columns 8 and 9 show call option prices and implied volatilities obtained by calibrating an MED to call and digital prices at all nine strikes.

Of course the MED obtained using calls and digitals at three strikes uses more ``information'' than the MED obtained using just calls at three strikes.
It is therefore not surprising that it leads to a better fit.
However, to show that this fit is already very good, we also report the MED obtained from call prices at all nine strikes used in this example.
Figure \ref{fig:5} illustrates graphically that these last two MED's are indeed very close to each other.

\begin{figure}[ht]
\includegraphics[width=\textwidth]{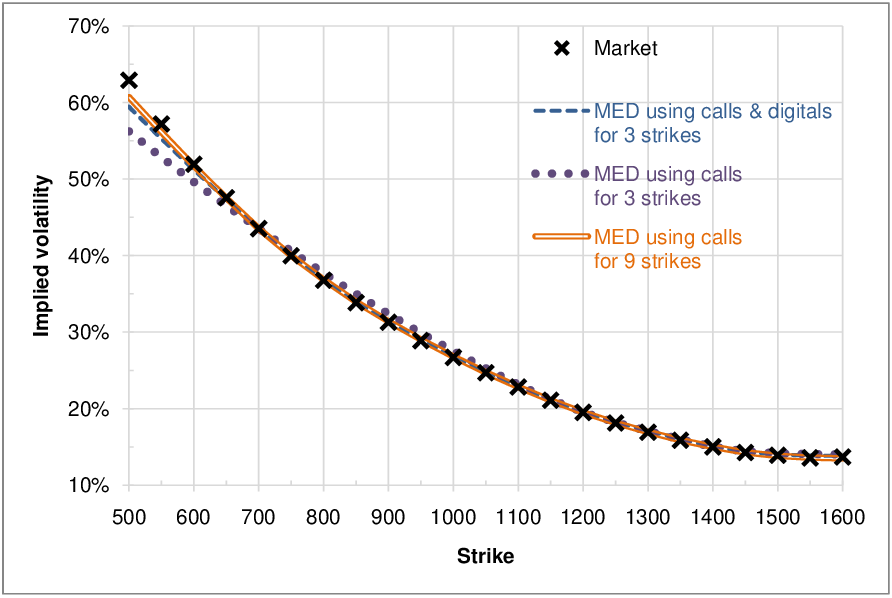}
\caption{Graphs of implied volatilities over strikes from Table \ref{tab:4}.}
\label{fig:5}
\end{figure}

\section{Some Remarks on Other Implied Distributions}

In most situations the information observed in the market regarding an asset consists of option prices
at a discrete set of strikes. Using this to extrapolate the second derivative of a function everywhere,
as suggested by Breeden and Litzenberger \cite{BL} or the volatility approach
\cite{DK}, \cite{D}, relies on additional assumptions about the distribution
of returns, the SDE the asset follows and/or the choice of an interpolation
method.
Even when there are strong reasons for such assumptions, we believe that it is important to know the shape
of the distribution function given by the Principle of Maximum Entropy (PME) in case these assumptions turn
out to be flawed.

In the local volatility model call prices for all strikes in $[0, \infty[$ are needed,
and, additionally, it assumes that the smile volatility is twice continuously differentiable.
Hence this approach requires an infinity of non-quoted prices together with a strong regularity.
``Since the market provides call prices at only a small number of strike prices, the second derivative
must be estimated by interpolation. This method is not very robust as the results are very sensitive
to the interpolation scheme used.'' \cite{BK}

For the method proposed here, if there are no observable digital quotes in the market,
the ``artificial'' data required consists only of digital prices for a finite, usually small, set of strikes ${\cal K} = \{K_1, ..., K_N\}$.
If one assumes (and our approach does not) that the volatility smile is differentiable with respect to the strike at the points of $\cal{K}$,
then prescribing digital prices there is indeed equivalent to prescribing the value of the smile derivative at those points.
This is still a much weaker requirement than that of the local volatility model.
Moreover, the example in Section \ref{SPX-Calibration} shows that centered call spread prices are very good estimators for digital prices.

\section{Conclusion}
\label{Conclusion}

Entropy has been one of the main concepts in information theory
\cite{S}, and since market participants react
to information when taking their positions, we believe entropy is a very natural tool for use in finance.

Entropy can be seen as a measure of how unbiased a probability distribution is.
Hence, by maximising entropy, what we propose is to find the most unbiased
probability distribution which agrees with information provided from the market.
We then show how this hypothesis leads to a piecewise exponential density.

The method we propose can be used reliably and efficiently in practice.
On the one hand, we have seen that it produces a remarkably realistic volatility surface from just one
volatility number as in the original Black-Scholes model, with a steep skew for short maturities that decays with increasing maturity.
On the other hand, if the actual distribution is known, then with option prices given at five or more strikes,
the fit to it is very close. In particular, it can be used as a robust interpolation method for volatility curves.

If additionally there is knowledge of a prior distribution, the Principle of Maximum Relative Entropy
can be applied to find a density that takes this into account and also meets the new constraints.
We give an example of such a scenario with two log-normal distributions, and show that the convergence to the actual
distribution is particularly quick.

Buchen and Kelly have proposed a similar method of finding a probability density that maximises entropy when
the market data consists only of call options. The density they obtain is continuous. However, to find its parameters they
must solve a multi-dimensional root-finding problem with the Newton-Raphson algorithm.

One criticism often raised in this application of the PME is that the method of finding the form of the density uses
Lagrange multipliers and is not rigorous. Indeed, this technique works well in practice and leads to the correct form,
but we also give a complete mathematical proof that avoids them.
Relative entropy has often been compared to a metric for probability distributions.
Our proof uses results by Csisz\'ar that give additional insights into ``distances'' between distributions
and establish remarkable ``geometric'' results.

Since we have an explicit form of the density, we are able to give analytical formulas
for the distribution, inverse distribution, and call and digital option prices.
Using Euler's relation for homogeneous functions, we give formulas for spot- and forward-deltas.

We also include two examples in which we calibrate to real market data from the CBOE.
We show that our method performs very well in both cases and compare our results to those obtained by using Buchen and Kelly's approach.

Cassio Neri currently works as quantitative analyst for Lloyds Banking Group in London. 
Lorenz Schneider is Assistant Professor in Quantitative Finance at EMLYON Business School in Lyon.
They have Ph.D.s in Mathematics from Universities Paris IX and VI, respectively.

\end{document}